\documentclass[10pt]{article}

\usepackage[utf8]{inputenc}
\usepackage{amsmath,amsthm,latexsym,amssymb,amsfonts,color}

\usepackage{authblk}
\usepackage{hyperref}
\usepackage{verbatim}

\newcommand{\bmat}[1]{\left(\begin{array}{cc}#1\end{array}\right)}

\newcommand{\R}{{\mathbb R}}
\newcommand{\Z}{{\mathbb Z}}
\newcommand{\C}{{\mathbb C}}
\newcommand{\z}{{\zeta}}

\newcommand{\LL}{L^2(\R,dx)}

\newcommand{\HF}{\widehat{\mathcal H}_g}
\newcommand{\HH}{{\mathcal H}_g}
\newcommand{\F}{\widehat{\mathcal F}}

\DeclareMathOperator{\dA}{\square}
\DeclareMathOperator{\supp}{supp}

\DeclareMathOperator{\x}{\mathfrak{x}_n}
\DeclareMathOperator{\xx}{\mathfrak{x}}

\newtheorem{thm}{Theorem}[section]
\newtheorem{lm}[thm]{Lemma}
\newtheorem{prop}[thm]{Proposition}
\newtheorem{df}[thm]{Definition}
\theoremstyle{remark}
\newtheorem{remark}[thm]{Remark}
\theoremstyle{remark}

\begin{document}

\title{Non-self-adjointness of the  Klein-Gordon operator on a globally hyperbolic and geodesically complete manifold. An example.}
\author[1]{Wojciech Kamiński}
\affil[1]{Faculty of Physics, University of Warsaw, Pasteura 5 PL-02093 Warsaw, Poland}

\maketitle

\begin{abstract}
We describe a Lorentzian manifold, that is globally hyperbolic and geodesically complete, but such that the (minimally coupled) Klein-Gordon operator with the standard domain is not essentially self-adjoint.
\end{abstract}

\section{Introduction}

Essential self-adjointness of the Laplacian (with a domain consisting of compactly supported smooth functions) is a quantum counterpart of completeness of classical motion i.e. geodesical completeness of the manifold. The known result of \cite{Strichartz} shows that for geodesically complete Riemannian manifolds Laplacian is essentially self-adjoint  (quantum complete). On the other hand, if the manifold is a part of a bigger space (metric extends smoothly through a smooth boundary), then the Laplace operator is not essentially self-adjoint.  Although there is no one-to-one correspondence (there are non-complete manifolds with self-adjoint Laplace operator), we see  that if we work only in the classically complete manifold setup we do not need to worry about corresponding quantum completeness property.

In the Lorentzian signature situation is more complicated. The counterpart of the Laplacian is the d'Alembert operator $\dA=-\nabla_\mu\nabla^\mu$ (see Section \ref{sec:dA}). It is formally symmetric in the Hilbert space of square integrable function with respect to the  natural measure $\mu_g$ associated to the metric. However, there is no natural notion of completeness and there are several proposed definitions which correspond to different (equivalent for Riemannian manifolds) characterizations of this property: geodesic completeness, timelike Cauchy completeness and finite compactness \cite{Global-lor}. One can argue that global hyperbolicity (see \cite{Ellis-Hawking, Global-lor}) is also related \cite{Global-lor}. For globally hyperbolic manifolds all previously mentioned properties are implied by geodesic completeness \cite{Beem-fin}. The question is thus if the self-adjointness is ensured by this property. Positive answer seems possible in light of  the Riemannian case. Indeed, in many special and important situations one can show that the d'Alembert $\dA$ as well as Klein-Gordon operator $\dA+m^2$ are essentially self-adjoint \cite{Derezinski2016, Vasy2017}. 

Let us notice, that self-adjointness would be very useful. First of all, it would legitimize formulas like $e^{is\dA}$ which sometimes appear in the literature about QFT in curved spacetimes (see for example the classical textbooks in the field \cite{Birrell} Section 6.1 and \cite{Parker2009} Section 3.4) although usually at formal level.\footnote{It is important to stress that such operators are not necessary in the theory of quantum fields on a curved background and they serve for derivation of so-called Hadamard parametrix, that can be obtained also by other completely rigorous methods (see \cite{Baer}). More importantly even if they exist, they might not be directly related to the Hadamard expansion. 
For example the kernel $e^{is\dA}$ on the circle $S^1$ is given by the distributional limit of the Jacobi's theta function, whereas all Hadamard coefficients are the same as on the real line (thus, for $n>0$ they are vanishing).} Secondly, with an additional assumption on the spectrum, by the limit of the inverses $(\dA+m^2\pm i\epsilon)^{-1}$ one should obtain distinguished Feynman and anti-Feynman propagators \cite{Derezinski2016}.

However, the answer is in the negative for the Lorentzian case. In this work, we will present an example of a Lorentzian spacetime (we will denote it by $M_V$) satisfying the following four conditions
\begin{enumerate}
\item it is diffeomorphic to $\R^4$,
\item it is geodesically complete,
\item it is globally hyperbolic,
\item the d'Alembert operator $\dA$, with the domain $C_0^\infty(M_V)$ of smooth, compactly supported functions, is not essentially self-adjoint in the Hilbert space $L^2(M_V,\mu_g)$  (see Section \ref{sec:dA-1} for the notation).
\end{enumerate}
Our example will be based on the example presented in \cite{Reed-Simon-2} concerning  similar question: whether self-adjointness follows from completeness of classical dynamics (for one dimensional particle moving in a potential).

\section{The example metric}\label{sec:metric}

In this article, we consider a spacetime $M_{V}$, that is a manifold $\R^4$ with a metric $g_{\mu\nu}$ of the form (signature $(-+++)$)\footnote{Although this metric has a form of the known $pp-$waves, it does not satisfy Einstein's equations with reasonable energy conditions. We thank D. Siemssen for pointing this fact to us.}
\begin{equation}
ds_g^2=-V(x)d\eta^2+2d\eta d\z+dx^2+dy^2,
\end{equation}
where $V(x)=-x^4+\sum_{n=1}^\infty (\sigma_n(x)+\sigma_n(-x))$ and $\eta,\z,x,y$ are  coordinates.

The functions $\sigma_n$ are assumed to have the following properties:
\begin{enumerate}
\item $\sigma_n$ is smooth and $\sigma_n\geq 0$,
\item $\supp \sigma_n\subset[\x,\x+\epsilon_n]$ where 
\begin{equation}\label{def-rn}
\x=\frac{1}{2}(n+1)+\frac{3}{2}\sqrt{n+1}
\end{equation}
and $0<\epsilon_n<\frac{1}{2}$ (such that $\x+\epsilon_n<\xx_{n+1}$ and $\sigma_n$ have disjoint supports) satisfy
\begin{equation}\label{def-epsn}
\sum_{n=1}^\infty \epsilon_n n^2<\infty,
\end{equation}
\item $\sup_{x\in\supp\sigma_n} \sigma_n(x)-x^4=n+1$.
\end{enumerate}

The functions $\sigma_n$ are basically the spikes introduced in \cite{Reed-Simon-2} (Example 2 on page 157). One can show (see also \ref{sec:sa-reduced}) that an operator on the real line $H=-\partial_x^2+p_{\z}^2 V(x)$ is not essentially self-adjoint on the standard domain $C_0^\infty(\R)\subset \LL$ if $p_{\z}\not=0$. Classical trajectory $x(s)$ for the particle in the potential without spikes, that is $V_0=-p_{\z}^2x^4$,  would reach plus or minus infinity in a finite time, so the trajectory cannot be extended indefinitely. On the other hand, in our case, classical motion is confined ($x(t)$ is bounded), because spikes are higher and higher and form a sort of barriers. This is the reason why maximal open interval of extendibility of a solution $t\rightarrow x(t)$ is the whole real line (we will say that the trajectory can be extended indefinitely).
We will use this property in our example.

For clarity of the presentation, we will now collect some basic properties of this metric:

\begin{lm}\label{lm:gg}
The following holds:
\begin{enumerate}
\item $\det g_{\mu\nu}=-1$.
\item The only nonzero components of the inverse metric $g^{\mu\nu}$ are
\begin{equation}
g^{\eta\z}=g^{\z\eta}=1,\ g^{\z\z}=V,\ g^{xx}=1,\ g^{yy}=1.
\end{equation}
In particular $d\eta$ is null: $g^{-1}(d\eta,d\eta)=0$.
\end{enumerate}
\end{lm}

\begin{proof}
In the matrix notation (order of the columns and rows $\eta,\z,x,y$) the metric is equal to
\begin{equation}
\left(\begin{array}{cccc}
-V & 1& 0 &0\\
1 & 0& 0 &0\\
0 & 0& 1 &0\\
0 & 0& 0 &1
\end{array}\right).
\end{equation}
The determinant is thus equal to $-1$ and the inverse metric is 
\begin{equation}
\left(\begin{array}{cccc}
0 & 1& 0 &0\\
1 & V& 0 &0\\
0 & 0& 1 &0\\
0 & 0& 0 &1
\end{array}\right).
\end{equation}
In the inverse metric $g^{\eta\eta}=0$, thus $d\eta$ is null.
\end{proof}

\section{Classical properties}

In this part of the paper we will review classical properties of the spacetimes. In our terminology, classical conditions are those which can be expressed through properties of geodesics. The specific part of this theory, tailored to features of Lorentzian geometry, is called causality theory and it is explained in \cite{Ellis-Hawking, Global-lor} and \cite{Minguzzi2019b}.

\subsection{Geodesic completeness}

Geodesic completeness means that every affinely parametrized geodesic can be extended indefinitely (\cite{Global-lor} Section 6.2). In other words, the 
maximal open interval of extendibility of any solution to the equation of affinely parametrized geodesics is the whole real line\footnote{Let us remind that the proper time cannot be defined for null geodesics, whereas for affine parametrization it is possible.}. In order to determine geodesics we will use Hamiltonian formulation. The geodesic equations for position $x^\mu(s)$ and corresponding momenta  $p_\mu(s)$ are given by
\begin{equation}
\dot{x}^\mu=\{x^\mu,h\},\ \dot{p}_\mu=\{p_\mu,h\},
\end{equation}
where $h=\frac{1}{2}g^{\mu\nu}(x)p_\mu p_\nu$, $\dot{}$ denotes the derivative with respect to affine parameter $s$, the Poisson bracket is given by $\{x^\mu,p_\nu\}=\delta^\mu_\nu$ and $g^{\mu\nu}$ is the inverse metric.\footnote{This Hamiltonian formulation for affine geodesics is based on the Lagrangian that can be found, for example, in \cite{Misner1973} Section 13.4.}

\begin{remark}
Let us notice, that operator $\frac{1}{2}\dA$ can be regarded as a quantization of the function $h$ \cite{Fulling-geo, DSL-geo}. This reinforces our terminology of classical and quantum properties, although different versions of quantizations can lead to modifications of $\dA$ for example by addition of a term proportional to the Ricci scalar \cite{Fulling-geo}.  In principle, such modifications can lead to essentially self-adjoint operators. We leave open the question, if they really change the situation.
\end{remark}

The classical Hamiltonian is in our case equal to
\begin{equation}
h=\frac{1}{2}g^{\mu\nu}p_\mu p_\nu=p_\eta p_{\z}+\frac{1}{2}V(x)p_{\z}^2+\frac{1}{2}p_x^2+\frac{1}{2}p_y^2,
\end{equation}
due to Lemma \ref{lm:gg}.
It has many conserved quantities $p_\eta$, $p_{\z}$, $p_y$ and
\begin{equation}\label{eq:CC}
C=p_x^2+V(x)p_{\z}^2.
\end{equation}
The function $C=2h-p_y^2-2p_{\z}p_\eta$ is conserved because $h$ is the \textit{time independent} Hamiltonian \cite{Arnold}.
Let us consider a geodesic with $p_{\z}\not=0$.
The equation \eqref{eq:CC} imposes restrictions on allowed values of $x(s)$
\begin{equation}
V(x(s))\leq \frac{C}{p_{\z}^2}
\end{equation}
There exist barriers with $V(x)> \frac{C}{p_{\z}^2}$ on the both sides of the initial value $x(0)$, because of the spikes in the potential. The trajectory cannot cross these barriers, so the motion in $x$ is bounded and $|x(s)|\leq D$ for some $D>0$. Let us introduce 
\begin{equation}
E=\sup_{|x|\leq D} |V(x)|.
\end{equation}
We can now estimate remaining velocities 
\begin{equation}
\dot{\eta}=p_{\z},\ |\dot{\z}|\leq|p_\eta|+E|p_{\z}|,\ \dot{y}=p_y.
\end{equation}
They are bounded and thus geodesics can be extended indefinitely (if they are parametrized by an affine parameter).

For $p_{\z}=0$ we have motion on the geodesics along a straight line, which can also be extended indefinitely.

\subsection{Time orientation}

Lorentzian geometry is special as it allows, in addition to the notion of standard orientation, to introduce also time orientation (\cite{Ellis-Hawking} Section 6.1, \cite{Minguzzi2019b} Section 1.7). 

A tangent vector $X$ is timelike if $g(X,X)<0$, null if $g(X,X)=0$ and nontrivial if $X\not=0$. The space of nontrivial, causal tangent vectors in a given point $q$ of the spacetime $M$  is defined as
\begin{equation}
\{X\in T_qM\colon g(X,X)\leq 0,\ X\not=0\}.
\end{equation}
It is a sum of two connected components which are related by multiplication by $-1$. This division depends continuously on points of the spacetime and locally we can decide in continuous way which part constitutes \textit{future directed causal vectors} and which \textit{past directed causal vectors}. If this can be done globally, we will say that spacetime is time-orientable. Not all spacetimes can be time oriented.\footnote{In such situation, there exists a double cover, that is time orientable \cite{Ellis-Hawking}. In our case $M_V$ is simply connected, so it is time-orientable.} However, time orientation is one of the crucial classical notions of physically reasonable spacetimes, that is at the bottom of so-called causality ladder (see \cite{Minguzzi2019b} Section 4 for thorough explanation).

Let us introduce a time orientation on $M_V$: A timelike vector $X=X^\eta\partial_\eta+X^z\partial_{\z}+X^x\partial_x+X^y\partial_y$ is future directed if
$X^\eta> 0$ and past directed if $X^\eta<0$. Let us remind that the scalar product of a null vector with a timelike vector is always nonzero. For every timelike vector $X$ indeed $X^\eta\not=0$, because $d\eta$ is null. We extend our definition by continuity to any nontrivial, causal vector.  

\begin{lm}\label{lm-1}
For $q\in M_V$, a causal vector $X\in T_qM_V$  is future directed if and only if one of the two excluding conditions holds
\begin{enumerate}
\item $X^\eta>0$,
\item $X^\eta=X^x=X^y=0$ and $X^{\z}<0$.
\end{enumerate}
\end{lm}

\begin{proof}
By continuity, if $X^\eta>0$ then the vector is future directed even if it is null. If $X^\eta=0$ then it needs to be null and
\begin{equation}
g(X,X)=(X^x)^2+(X^y)^2\leq 0\Rightarrow X^x=X^y=0\text{ and } X=X^{\z} \partial_{\z}.
\end{equation}
Let us consider small perturbation of this vector
\begin{equation}
X(\lambda)=\lambda \delta X^\eta \partial_\eta+X^z\partial_{\z}.
\end{equation}
Furthermore
\begin{equation}
g(X(\lambda),X(\lambda))=2\lambda\delta X^\eta X^z+O(\lambda^2),
\end{equation}
so if $X^{\z}\delta X^\eta<0$ then $X(\lambda)$ is timelike for $\lambda\in (0,\lambda_{max})$ with some $\lambda_{max}>0$. It is also future directed if $\delta X^\eta>0$.  We see that if $X^{\z}<0$ then there exists a one parameter family of timelike future directed vectors converging to $X$, thus $X$ is future directed. If $X^{\z}>0$ then similar considerations shows that it is past directed.
\end{proof}

\begin{remark}
The potential $V(x)$ does not have a definite sign. If $V(x)\leq 0$ then for a future directed causal vector $X$ it holds $X^{\z}<0$ independently whether $X^\eta=0$ or not. However, for points where $V(x)>0$ the sign of $X^{\z}$ can be arbitrary for future directed causal vectors.
\end{remark}

\subsection{Causality}

Once we know that a spacetime admits time orientation, the next important property is causality (see \cite{Ellis-Hawking} Section 6.4, \cite{Global-lor} Section 3.2). We will explain now this concept. For simplicity we will work with the following class of curves (see \cite{Minguzzi2019b} Section 2.3):

\begin{df}
A future (respectively past) directed, causal, absolutely continuous (abbreviation AC) curve is an absolutely continuous map
\begin{equation}
\gamma\colon [0,1]\rightarrow M,
\end{equation}
such that derivative $\dot{\gamma}$ is future (respectively past) directed causal vector for almost every $s\in [0,1]$.
\end{df}

Here $\dot{}$ denotes derivative with respect to $s$. In our terminology future or past directed causal vectors are always nontrivial.

\begin{remark}
Absolutely continuous in this context means, that all coordinate functions are absolutely continuous: their derivatives exist almost everywhere, the derivatives are Lebesgue integrable and the fundamental theorem of calculus holds in terms of the Lebesgue integral (see \cite{Rudin-RC}, Chapter 7). 
\end{remark}

\begin{remark} \label{rmk:1}
As shown in \cite{Minguzzi2019b} Theorem 2.12 every two points connected by a future directed AC curve $\gamma$ can be also connected by a curve, that consists of the segments of geodesics and this curve can be chosen belonging to an arbitrary small neighbourhood of $\gamma([0,1])\subset M$. The same is true also for so-called continuous causal curves \cite{Minguzzi2019b}, thus there is a freedom in the choice of the class of curves, that will not affect constructions and definitions of the objects considered in this paper.
\end{remark}

The \textit{causal spacetimes} are defined by the property: there are no closed, future directed causal AC curves. These are spacetimes without time machine paradoxes.

\begin{prop}\label{lm-c}
Suppose, that there exists in $M_V$ a future directed causal AC curve from a point $q_0=(\eta_0,\z_0,x_0,y_0)$ to a point $q_1=(\eta_1,\z_1,x_1,y_1)$ then one of two mutually excluding conditions holds
\begin{enumerate}
\item $\eta_1>\eta_0$,
\item $\eta_1=\eta_0$ and $\z_1<\z_0$.
\end{enumerate}
In particular $p_1\not=p_0$ and  there are no closed, causal AC curves.
\end{prop}

\begin{proof}
Let us consider a future directed causal AC curve
\begin{equation}
[0,1]\ni s\rightarrow \gamma(s)=(\eta(s),\z(s),x(s),y(s)),\quad \gamma(0)=q_0,\ \gamma(1)=q_1.
\end{equation}
Along the curve $\dot{\eta}\geq 0$ (almost everywhere) by Lemma \ref{lm-1}, thus we have $\eta_1\geq \eta_0$. Moreover, if $\eta_1=\eta_0$ then on the curve $\dot{\eta}=0$ and so $\dot{\z}<0$ (also almost everywhere) and $\z_1<\z_0$. We conclude that $q_1\not=q_0$ and there are no closed, future directed causal AC curves.
\end{proof}

\subsection{Global hyperbolicity}\label{sec:comp}

We will now show  the important, classical property of the spacetime, namely, global hyperbolicity. The spacetime is \textit{globally hyperbolic} if any of the two equivalent conditions hold\footnote{There are many equivalent characterizations of globally hyperbolic spacetimes \cite{Minguzzi2019b}.}
\begin{enumerate}
\item The spacetime is diffeomorphic to $\R\times \Sigma$ where for every $t\in\R$, $\{t\}\times \Sigma$ is a Cauchy surface (\cite{Bernal2003, Global-lor,Ellis-Hawking} and \cite{Minguzzi2019b} Theorem 4.120)
\item Spacetime is strongly causal and closure of every causal diamond is compact
(see \cite{Minguzzi2019b} Section 4.5)\footnote{In \cite{Minguzzi2019b} a weaker property of non-imprisoning is assumed, but it follows from strong causality (see \cite{Minguzzi2019b}). Actually, if closures of all causal diamonds are compact, these two conditions are equivalent, because every globally hyperbolic spacetime according to the definition from \cite{Minguzzi2019b} is also strongly causal.}
\end{enumerate}
We will use the second definition that we will now explain in details. 

The $+$ future/ $-$ past developments $J^\pm(q)$ of a point $q$ are defined as follows
\begin{align}
J^\pm(q)=&\{q'\in M\colon \exists \text{ a future/past directed causal AC curve}\nonumber\\ &\gamma\colon [0,1]\rightarrow M,\ \gamma(1)=q',\ \gamma(0)=q\}\cup \{q\}.
\end{align}
According to Remark \ref{rmk:1} we can replace our definition of causal curves by various other types of curves from \cite{Minguzzi2019b} without changing the resulting sets.

The \textit{causal diamond} for two points $q_0,q_1\in M$ is the set
\begin{equation}
J^+(q_0)\cap J^-(q_1),
\end{equation}
The spacetime is \textit{strongly causal} if for every point $q\in M$ and every open neighbourhood $U$ of $q$, there exist points $q_0,q_1\in M$ such that
\begin{equation}
q\in J^+(q_0)\cap J^-(q_1)\subset U.
\end{equation}
As the name suggests, it is a stronger version of causality \cite{Minguzzi2019b}.

Let us stress out why global hyperbolicity is important. Globally hyperbolic spacetimes have a well-posed Cauchy problem for the Klein-Gordon equation (as well as they enjoy other nice properties, see \cite{Ellis-Hawking, Global-lor}). 
This property is of utmost importance in the formulation of quantum field theory on the curved background \cite{Wald}. The class of globally hyperbolic spacetimes is believed to be  the class in which the QFT can be reasonably formulated \cite{Wald}. 

In order to prove that $M_V$ is globally hyperbolic we will need few technical results:\footnote{We have chosen to work with this definition to avoid technicalities related to limits of causal AC curves (see \cite{Minguzzi2019b}).}

\begin{lm}\label{lm-2}
Let $q=(\eta,\z,x,y)\in M_V$ be such that $|x|\leq \x$ (see Equation \eqref{def-rn} for the definition of $\x$). Then for a causal future directed vector $X\in T_qM_V$ the following inequalities hold
\begin{enumerate}
\item $\sqrt{n}X^{\eta}-\frac{1}{\sqrt{n}}X^{\z}> 0$,
\item $|X^x|\leq \sqrt{n}X^\eta-\frac{1}{\sqrt{n}}X^{\z}$,
\item $|X^y|\leq \sqrt{n}X^\eta-\frac{1}{\sqrt{n}}X^{\z}$,
\item $|X^{\z}|\leq nX^\eta-X^{\z}$.
\end{enumerate}
\end{lm}

\begin{proof}
For such $q$ we have  (as $V(x)\leq n$)
\begin{align}
&0\geq g(X,X)\geq -n (X^\eta)^2+2X^\eta X^{\z}+(X^x)^2+(X^y)^2=\nonumber\\
&=-\left(\sqrt{n}X^\eta-\frac{1}{\sqrt{n}}X^{\z}\right)^2+\frac{1}{n}(X^{\z})^2+
(X^x)^2+(X^y)^2. \label{eq-X}
\end{align}
Thus for a causal vector $\sqrt{n}X^\eta-\frac{1}{\sqrt{n}}X^{\z}\not=0$. In the given point the cone of future directed causal directions is connected, thus $\sqrt{n}X^\eta-\frac{1}{\sqrt{n}}X^{\z}$ has in it a definite sign. We can determine this sign by taking a single future directed causal vector, for example $-\partial_{\z}$. In this case $X^{\z}=-1$ and other components vanish thus $\sqrt{n}X^\eta-\frac{1}{\sqrt{n}}X^{\z}>0$.

From \eqref{eq-X} we obtain for any causal future directed vector $X$
\begin{equation}
|X^x|\leq \left|\sqrt{n}X^\eta-\frac{1}{\sqrt{n}}X^{\z}\right|=\sqrt{n}X^\eta-\frac{1}{\sqrt{n}}X^{\z}.
\end{equation}
The same applies for $X^y$ and a similar inequality holds for $X^{\z}$.
\end{proof}

We will now state the second technical fact:

\begin{lm}\label{lm:bounded}
Let us consider a future directed causal AC curve
\begin{equation}
[0,1]\ni s\rightarrow \gamma(s)=(\eta(s),\z(s),x(s),y(s))\in M_V.
\end{equation}
Let us denote $\Delta\eta=\eta(1)-\eta(0)$, $\Delta\z=\z(1)-\z(0)$ and let us choose $N\in \Z_+$ such that
\begin{equation}\label{eq:NNN}
|x(0)|\leq \sqrt{N},\quad |x(1)|\leq \sqrt{N},\quad \Delta\eta\leq \sqrt{N},\quad \Delta \z\geq -\sqrt{N}.
\end{equation}
Then the following inequalities hold for any $s\in [0,1]$
\begin{align}
|x(s)-x(0)|&\leq\sqrt{N}\Delta\eta-\frac{1}{\sqrt{N}}\Delta{\z}\\
|y(s)-y(0)|&\leq\sqrt{N}\Delta\eta-\frac{1}{\sqrt{N}}\Delta{\z}\\
|\z(s)-\z(0)|&\leq N\Delta\eta-\Delta{\z}\\
|\eta(s)-\eta(0)|&\leq \Delta\eta
\end{align}
\end{lm}

\begin{proof}
First we show that $|x(s)|<\mathfrak{x}_N$ for all $s\in[0,1]$. Define 
\begin{equation}
n:=\min\left\{m\in \Z_+\cup\{0\}\colon\sup_{s\in[0,1]} |x(s)|< \xx_m\right\}.
\end{equation}
The supremum is finite (as the interval $[0,1]$ is compact), so the set of $m$ is not empty.

We assume first that $n>0$.
Thus, there exists $s\in[0,1]$ such that $|x(s)|\geq\mathfrak{x}_{n-1}$.  On one hand, we can estimate (as $x(t)$ is absolutely continuous)
\begin{align}
& \int_0^1|\dot{x}(t)|dt= \int_0^s|\dot{x}(t)|dt+ \int_s^1|\dot{x}(t)|dt\geq|x(s)-x(0)|+|x(s)-x(1)|\geq\nonumber\\
&\geq 2|x(s)|-|x(0)|-|x(1)|\geq n+3\sqrt{n}-2\sqrt{N}
\end{align}
using $|x(s)|\geq\mathfrak{x}_{n-1}$ and \eqref{eq:NNN}
in the last inequality.
On the other hand, from inequalities in Lemma \ref{lm-2}, which are valid for almost every $t$, we derive
\begin{align}
&\int_0^1|\dot{x}(t)|dt\leq
\int_0^1 \left(\sqrt{n}\dot{\eta}(t)-\frac{1}{\sqrt{n}}\dot{z}(t)\right)dt=\nonumber\\
&=\sqrt{n}\Delta\eta-\frac{1}{\sqrt{n}}\Delta{\z}\overset{\lozenge}{\leq} \left(\sqrt{n}+\frac{1}{\sqrt{n}}\right)\sqrt{N}\leq \left(\sqrt{n}+1\right)\sqrt{N}.
\end{align}
where we used \eqref{eq:NNN} in the $\lozenge$ inequality.
Thus
\begin{equation}
n+3\sqrt{n}-2\sqrt{N}\leq \left(\sqrt{n}+1\right)\sqrt{N}\ \Longrightarrow
\left(\sqrt{n}+3\right)\sqrt{n}\leq \left(\sqrt{n}+3\right)\sqrt{N},
\end{equation}
so we proved that $n\leq N$. The last conclusion is also true if $n=0$. 

Independently of the choice of the curve $\gamma$
\begin{equation}
|x(s)|<\mathfrak{x}_{N},
\end{equation}
because $\xx_n\leq\xx_N$ for $n\leq N$.

We can now use Lemma \ref{lm-2} to estimate 
\begin{align}
&|x(s)-x(0)|\leq \int_0^s |\dot{x}(t)|dt\leq \int_0^1 |\dot{x}(t)|dt\leq\int_0^1 \left(\sqrt{N}\dot{\eta}(t)-\frac{1}{\sqrt{N}}\dot{\z}(t)\right)dt=\nonumber\\
&=\sqrt{N}\Delta\eta-\frac{1}{\sqrt{N}}\Delta{\z},
\end{align}
for any parameter $s\in[0,1]$.
Exactly the same estimate holds also for $y(s)$. Namely,
\begin{equation}
|y(s)-y(0)|\leq \sqrt{N}\Delta\eta-\frac{1}{\sqrt{N}}\Delta{\z}.
\end{equation}
Similarly,
\begin{align}
|\z(s)-\z(0)|&\leq \int_0^s |\dot{\z}(t)|dt\leq \int_0^1 |\dot{\z}(t)|dt\leq\int_0^1\left( N\dot{\eta}(t)-\dot{\z}(t)\right)dt=N\Delta\eta-\Delta{\z}.
\end{align}
Finally, the function $\eta$ is non-decreasing along the causal future directed AC curve. Thus
\begin{equation}
\eta(s)-\eta(0)\leq \Delta\eta,
\end{equation}
that shows the last inequality.
\end{proof}

We are now ready to prove our result:

\begin{thm}\label{thm-gl}
The spacetime $M_V$ is globally hyperbolic.
\end{thm}

\begin{proof}
Let us first show that $M_V$ is strongly causal. Choose 
\begin{equation}
q=(\eta,\z,x,y)\subset U.
\end{equation}
We introduce a notation for every $\beta>0$
\begin{equation}
O_q(\beta)=[\eta-\beta,\eta+\beta]\times 
[\z-\beta,\z+\beta]\times 
[x-\beta ,x+\beta]\times 
[y-\beta,y+\beta ]\subset M_V.
\end{equation}
There exist $\delta>0$ such that
\begin{equation}
O_q(2\delta)\subset U.
\end{equation}
Let us now choose $N\in \Z_+$ such that 
\begin{equation}\label{eq:Nmax}
N>(|x|+2\delta)^2
\end{equation}
and set $\epsilon=\frac{\delta}{2(N+1)}$. Next we choose two points
\begin{equation}
q_i=(\eta_i,\z_i,x_i,y_i),\quad i=0,1,
\end{equation}
such that $q_0\in J^-(q)$, $q_1\in J^+(q)$ and
\begin{equation}
q_0,q_1\in O_q(\epsilon).
\end{equation}
Using \eqref{eq:Nmax} we show
\begin{equation}
|x_0|\leq \sqrt{N},\ |x_1|\leq \sqrt{N},\ \eta_1-\eta_0\leq \sqrt{N},\ \z_1-\z_0\geq -\sqrt{N},
\end{equation}
where the last inequalities follows from the fact that $2\delta\leq \sqrt{N}$.
Recall Lemma \ref{lm:bounded}. For a future directed causal AC curve from $q_0$ to $q_1$:
\begin{equation}
[0,1]\ni s\rightarrow \gamma(s)=(\eta(s),\z(s),x(s),y(s)), \quad \gamma(0)=q_0,\ \gamma(1)=q_1,
\end{equation}
we have
\begin{equation}
|x(s)-x_0|\leq \sqrt{N}(\eta_1-\eta_0)-\frac{1}{\sqrt{N}}(\z_1-\z_0)\leq 
\left(\sqrt{N}+\frac{1}{\sqrt{N}}\right)2\epsilon\leq \delta
\end{equation}
Similarly $|y(s)-y_0|\leq \delta$, $|\z(s)-\z_0|\leq \delta$ and $|\eta(s)-\eta_0|\leq \delta$. Hence, for every $s$
\begin{equation}
\gamma(s)=(\eta(s),\z(s),x(s),y(s))\in O_q(2\delta)\subset U.
\end{equation}
As the causal AC curve was arbitrary, we have $J^+(q_0)\cap J^-(q_1)\subset U$. This proves strong causality of $M_V$.

For any two points $q_0,q_1\in M_V$ all coordinates of $J^+(q_0)\cap J^-(q_1)$ are bounded thus
\begin{equation}
\overline{J^+(q_0)\cap J^-(q_1)}
\end{equation}
is compact. This shows global hyperbolicity of the spacetime $M_V$.
\end{proof}

\section{Non-self-adjointness}\label{sec:dA-1}

We will now turn to \textit{quantum properties} of spacetimes. These are properties related to propagation of the fields, which are not expressed simply in terms of properties of geodesics. The simplest example concerns properties of the Klein-Gordon equation on a spacetime $M$
\begin{equation}
(\dA+m^2)\phi=0,
\end{equation}
for the scalar field with mass $m$. The d'Alembert operator $\dA$ is defined by\footnote{Here, for simplicity we use Einstein summation convention.}
\begin{equation}
\dA=-g^{\mu\nu}\nabla_\mu\nabla_\nu,
\end{equation}
where $g^{\mu\nu}$ is the inverse metric and $\nabla_\mu$ is the covariant derivative.

There exists a distinguished measure defined by the Lorentzian metric 
\begin{equation}
\mu_g=\sqrt{|\det g_{\mu\nu}|}d^4x.
\end{equation}
Thus, as in the Riemannian case, we can define a natural Hilbert space structure
\begin{equation}
\HH=L^2(M,\mu_g).
\end{equation}
The subspace $C^\infty_0(M)$ is dense in this Hilbert space.
The d'Alembert operator (as well as Klein-Gordon operator, that differs by a constant) is formally Hermitian. Precisely, if we consider $\dA$ as an operator with the domain ${\mathcal D}=C^\infty_0(M)$ (smooth and compactly supported functions) then\footnote{We use notation for the scalar product $\langle\cdot,\cdot\rangle$ and the norm $\|\cdot\|$. The scalar product is anti-linear in the first variable.}
\begin{equation}
\langle \psi, \dA\phi\rangle=\langle \dA\psi, \phi\rangle,\quad \psi,\phi\in {\mathcal D}.
\end{equation}
This is not enough to ensure essential self-adjointness (see \cite{Reed-Simon-2} Section X.1) and it might happen, that the formulas like $e^{it\dA}$ are ambiguous.\footnote{There exists at least one self-adjoint extension, because $\dA$ is real \cite{Reed-Simon-2}.} 

We will now show, that exactly this is the case for $M_V$. 

\begin{thm}\label{thm:2}
On the spacetime $M_V$, the d'Alembert operator $\dA$ with the
domain $C^\infty_0(M_V)\subset \HH=L^2(M_V,\mu_g)$ is not essentially self-adjoint.
\end{thm}

The same is true for the Klein-Gordon operator $\dA+m^2$, $m\in\R$, since it differs only by a constant.
The proof of this theorem will occupy the rest of the work.

\subsection{Reformulation of the problem}\label{sec:dA}

On the $M_V$ spacetime  $\det g_{\mu\nu}=-1$ holds, so we have $\mu_g=d\eta d\z dx dy$.
Let us notice the identity (\cite{Birrell} Section 2.2) satisfied by the d'Alembert operator\footnote{We use Einstein's summation convention.} 
\begin{equation}
\dA\psi=-\frac{1}{\sqrt{g}}\partial_\mu (\sqrt{g}g^{\mu\nu} \partial_\nu\psi),
\end{equation}
where $\sqrt{g}=\sqrt{|\det g_{\mu\nu}|}$. As the determinant $\det g_{\mu\nu}=-1$ the operator takes the form
\begin{equation}
\dA=-2\partial_\eta\partial_{\z}-V(x)\partial_{\z}^2-\partial_x^2-\partial_y^2.
\end{equation}
Let us now consider a Hilbert space of functions $\phi(x,p_y,p_{\z},p_\eta)$
\begin{equation}
\HF=L^2(\R^4,\hat{\mu}_g),\quad \hat{\mu}_g=dx dp_ydp_{\z}dp_{\eta}.
\end{equation}
The scalar product of the functions $\phi_i\in \HF$ is given by
\begin{equation}
\langle\phi_1,\phi_2\rangle_{\hspace{-1pt}\wedge}=\int dx dp_y dp_{\z} dp_\eta\ 
\overline{\phi_1(x,p_y,p_{\z},p_\eta)}\phi_2(x,p_y,p_{\z},p_\eta).
\end{equation}
We now introduce a unitary operator which is a partial Fourier transform
\begin{equation}
\F\colon \HH \rightarrow \HF,
\end{equation}
given by the formula for $\psi\in \HH$
\begin{equation}\label{eq:F}
\F(\psi)(x,p_y,p_{\z},p_\eta)=\frac{1}{(2\pi)^{3/2}}\int_{\R^3} d\eta d\z dy \ \psi(\eta,\z,x,y)e^{-i(p_{\eta}\eta +p_{\z}\z+p_y y)}.
\end{equation}
We will denote by $\widehat{\dA}$ our operator in this representation 
\begin{equation}
\widehat{\dA}=\F\circ \dA \circ \F^{-1},
\end{equation}
with the domain $\widehat{D}=\F(C^\infty_0(M_V))\subset \HF$.

Let us define a family of operators on the domain $C^\infty_0(\R)\subset \LL$
\begin{equation}
\dA_{p_y,p_{\z},p_\eta}=-\partial_x^2+V(x)p_{\z}^2+2p_\eta p_{\z}+p_y^2.
\end{equation}
The partial Fourier transform \eqref{eq:F} block diagonalizes our operator $\dA$ in the following sense:

\begin{lm}\label{lm:F}
Let $\phi\in \widehat{D}$, then $\phi\in C^\infty(\R^4)$ (smooth functions) and there exists $L>0$ such that
\begin{equation}
\supp \phi\subset[-L,L]\times \R^3.
\end{equation}
In particular, functions $\phi_{p_y,p_{\z},p_\eta}(x)=\phi(x,p_y,p_{\z},p_\eta)$ are smooth and of compact support and moreover
\begin{equation}\label{eq:diag}
[\widehat{\dA}\phi](x,p_y,p_{\z},p_\eta)=[\dA_{p_y,p_{\z},p_\eta}\phi_{p_y,p_{\z},p_\eta}](x).
\end{equation}
\end{lm}

\begin{proof}
Let us notice that $\phi=\F(\psi)$ for $\psi\in C^\infty_0(M_V)$. There exists $L>0$  such that
\begin{equation}
\forall_{|x|>L}\psi(\eta,\z,x,y)=0,
\end{equation}
thus $\phi(x,p_y,p_{\z},p_\eta)=0$ for $|x|>L$ and the functions $\phi_{p_y,p_{\z},p_\eta}$ are compactly supported. Function $\phi$ is smooth as the Fourier transform of smooth compactly supported function. From that also follows, that functions $\phi_{p_y,p_{\z},p_\eta}$ are smooth, thus $\phi_{p_y,p_{\z},p_\eta}\in C^\infty_0(\R)$.
The equation \eqref{eq:diag} is a basic property of the Fourier transform. 
\end{proof}

\subsection{Strategy of the proof}


W describe now the strategy of the proof of Theorem \ref{thm:2}.

First part of the proof will be devoted to the analysis of $\dA_{p_y,p_{\z},p_\eta}$ with the domain $C^\infty_0(\R)$. These are second order differential operators on the real line and the theory of such operators is highly developed (see references in \cite{Reed-Simon-2}). They turn out to be not essentially self-adjoint. We will show, that both kernels $\ker (\dA_{p_y,p_{\z},p_\eta}^\dagger\pm i)$ are two dimensional. 
Following standard notation (see \cite{Reed-Simon-2} Section X.1) we will say in this situation, that \textit{the operator has both deficiency indices equal to $2$}.
As we mentioned in Section \ref{sec:metric}, our example is based on specific counter-example from \cite{Reed-Simon-2}. We will sketch how this property follows from \cite{Reed-Simon-2}. However, we will provide self-consistent proof of this fact in Lemma \ref{lm-8}.

In Section \ref{sec:non} we will show how to construct a function in $\HF$ from functions in the kernels $\ker (\dA_{p_y,p_{\z},p_\eta}^\dagger+ i)$ in such a way, that it is measurable and the norm is finite. Applying Lemma \ref{lm:F}, we will show, that our function is perpendicular to the subspace
\begin{equation}\label{eq:imd}
\{\psi\in \HF\colon \psi=(\widehat{\dA}-i)\phi,\ \phi\in \widehat{D}\},
\end{equation}
so it belongs to $\ker (\widehat{\dA}^\dagger+ i)$ (see \cite{Reed-Simon-2} Section X.1). That finishes the proof.

\subsection{Non-self-adjointness of the reduced operator}
\label{sec:sa-reduced}

According to \cite{Reed-Simon-2} (Theorem X.9 and Example 3 following after it) for every values of $p_\eta,p_{\z},p_y$, where $p_{\z}\not=0$  the operator with the domain $C^\infty_0(\R)\subset \LL$ (acting in $x$ and parametrized by $p_\eta,p_{\z},p_y$)
\begin{equation}
H_0=-\partial^2_x-p_{\z}^2x^4+p_y^2+2p_\eta p_{\z}
\end{equation}
has both deficiency indices equal to $2$. This means, in this case, that for any $\lambda\in\C$ there exist two dimensional family of distributional solutions to the equation
\begin{equation}
(H_0^\dagger+\lambda)\psi=0,\quad \psi\in D'(\R),
\end{equation}
which are square integrable. In fact, every such function is smooth as the potential is smooth.

In order to prove the same defficiency indices for $\dA_{p_y,p_{\z},p_\eta}=H_0+V_1$ it is then enough to show that the additional potential 
\begin{equation}
V_1=p_{\z}^2\sum_{n=1}^\infty (\sigma_n(x)+\sigma_n(-x))
\end{equation}
is $H_0$-bounded with a bound less than one (\cite{Reed-Simon-2} Theorem X.12). This can be done in the way as in \cite{Reed-Simon-2} page 158 \footnote{The verbatim application of the method from \cite{Reed-Simon-2} will show the result only for sufficiently small $p_{\z}$. We will not dwell into details, as we provide an alternative proof valid for all $p_{\z}\not=0$.}. However,  we provide here an alternative proof of the lack of essential self-adjointness for the operator $\dA_{p_y,p_{\z},p_\eta}$ (based on a method from \cite{Wintner}) for $p_{\z}\not=0$. We will use the following criterium:

\begin{lm}\label{lm-pr}
Let us suppose that $V_o, V_b\in C^\infty (\R)$ satisfy 
\begin{enumerate}
\item \label{en:1} $V_o\leq -C$, for some $C>0$,
\item \label{en:2} $(-V_o)^{-\frac{1}{2}}\in L^1(\R,dx)$,
\item \label{en:3} $(-V_o)^{-\frac{1}{2}}\left(\frac{5(V_o')^2-4V_o''V_o}{16V_o^2}-V_b\right)\in L^1(\R,dx)$,
\end{enumerate}
here $'$ denotes derivative.
Then the operator $H=-\partial_x^2+V_o+V_b$ with the domain $C^\infty_0(\R)\subset \LL$ has both deficiency indices equal to $2$.
\end{lm}


\begin{proof}
We will show that two independent distributional solutions
\begin{equation}\label{sol:Ham}
(H^\dagger+\lambda) \phi=0
\end{equation}
are square integrable for any $\lambda\in \C$. 

The idea of the proof is to show, that the so-called WKB or Liouville-Green \cite{Olver1997} ansatz $\Phi_\pm=\frac{1}{\sqrt{2S'}}e^{\pm i\int_0^x S'}$ (where $S'=(-V_o)^{\frac{1}{2}}$) approximates well enough solutions to the differential equation \eqref{sol:Ham}. 
In order to show the approximation property we will rewrite the differential equation in the matrix, but first order form:
\begin{equation}\label{sol:Ham-matrix}
\frac{d}{dx}\bmat{\phi_+' & \phi_-'\\ \phi_+ &\phi_-}=\bmat{0 & \lambda+V_o+V_b\\ 1&0}\bmat{\phi_+' & \phi_-'\\ \phi_+ &\phi_-}.
\end{equation}
We introduce an auxiliary matrix
\begin{equation}
\underbrace{\bmat{\Psi_+ & \Psi_-\\ \Phi_+ &\Phi_-}}_{=M}=\frac{1}{\sqrt{2S'}}\underbrace{\bmat{ -\frac{1}{2}\frac{S''}{S'}+iS' & -\frac{1}{2}\frac{S''}{S'}-iS' \\ 1 & 1}}_{=A}\underbrace{\bmat{e^{iS} &0\\ 0 &e^{-iS}}}_{=B},
\end{equation}
where $S(x)=\int_0^x (-V_o)^{\frac{1}{2}}$ and $'$ denotes derivative with respect to $x$. The functions $\Phi_\pm$ are square integrable by condition \ref{en:2} and the matrix $M$ satisfies
\begin{equation}\label{sol:matrix-M}
M'=\bmat{0 & V_o+V_p\\ 1 &0}M,
\end{equation}
where $V_p=\frac{5(V_o')^2-4V_o''V_o}{16V_o^2}$. The inverse of $M$ is
\begin{equation}
M^{-1}=-iB^{-1}\frac{1}{\sqrt{2S'}}\underbrace{\bmat{1 & \frac{1}{2}\frac{S''}{S'}+iS'\\ -1 & -\frac{1}{2}\frac{S''}{S'}+iS'}}_{=C}.
\end{equation}
We can now write the solution to \eqref{sol:Ham-matrix} in the form
\begin{equation}
\bmat{\phi_+' & \phi_-'\\ \phi_+ &\phi_-}=M U,
\end{equation}
where the invertible matrix $U$ satisfies by \eqref{sol:Ham-matrix} and \eqref{sol:matrix-M}
\begin{equation}
U'=M^{-1}\bmat{0 & \lambda+V_b-V_p\\ 0 &0}M U.
\end{equation}
We will show that $\|M^{-1}\bmat{0 & \lambda+ V_b-V_p\\ 0 &0}M\|\in L^1(\R,dx)$ ($l^2$ matrix norm) thus $U$ has a limit in $\pm\infty$ and, as $\Phi_\pm$ are square integrable, the same is true for $\phi_\pm$ because
\begin{equation}
\bmat{\phi_+ &\phi_-}=\bmat{\Phi_+ &\Phi_-}U.
\end{equation}
In fact, it is enough to show, that
\begin{equation}
\left\|\frac{1}{2S'}C\bmat{0 & \lambda+V_b-V_p\\ 0 &0}A\right\|\in L^1(\R,dx),
\end{equation}
as $B$ is unitary. However, we see that the latter matrix is just
\begin{equation}
\frac{\lambda+V_b-V_p}{2S'}\bmat{1 &1\\ -1 & -1}=\frac{\lambda}{2S'}\bmat{1 &1\\ -1 & -1}+\frac{V_b-V_p}{2S'}\bmat{1 &1\\ -1 & -1}.
\end{equation}
The conditions \ref{en:2} and \ref{en:3} for the potential ensures integrability, thus the operator has both deficiency indices equal to $2$.
\end{proof}

We can now apply this lemma to $\dA_{p_y,p_{\z},p_\eta}$:

\begin{lm}\label{lm-8}
The operator $\dA_{p_y,p_{\z},p_\eta}$ for $p_{\z}\not=0$ with the domain $C_0^\infty(\R)\subset \LL$ is not essentially self-adjoint and it has both deficiency indices equal to $2$.
\end{lm}

\begin{proof}
Let us consider operator
\begin{equation}\label{eq:H_1}
H_1=\dA_{p_y,p_{\z},p_\eta}-p_{\z}^2-2p_\eta p_{\z}-p_y^2=-\partial^2_x-p_{\z}^2(x^4+1)+V_1,
\end{equation}
with $V_1=p_{\z}^2\sum_{n=1}^\infty (\sigma_n(x)+\sigma_n(-x))$ and domain $C_0^\infty(\R)\subset \LL$. We will apply now Lemma \ref{lm-pr}. Choosing $V_o=-p_{\z}^2(x^4+1)$, $V_b=V_1$, $C=p_{\z}^2$ we can estimate
\begin{enumerate}
\item $V_o\leq -C$.
\item $(-V_o)^{-\frac{1}{2}}=\frac{1}{|p_{\z}|\sqrt{x^4+1}}\in L^1(\R,dx)$.
\item From the inequality $|2x^6-3x^2|\leq 3(x^4+1)^{\frac{3}{2}}$ we obtain 
\begin{equation}
\left|\frac{5(V_o')^2-4V_o''V_o}{16V_o^2(-V_o)^{\frac{1}{2}}}\right|=\frac{\left|5p_{\z}^4x^6-3p_{\z}^4 x^2(x^4+1)\right|}{|p_{\z}|^5(x^4+1)^{\frac{5}{2}}}\leq 
\frac{3}{|p_{\z}|(x^4+1)}\in L^1(\R,dx).
\end{equation}
\item $(-V_o)^{-\frac{1}{2}}V_b\in L^1(\R,dx)$: Let us notice,  that 
\begin{equation}
\sup_{x\in \supp\sigma_n} \sigma_n(x)\leq \sup_{x\in \supp\sigma_n} x^4+n+1\leq (\x+\epsilon_n)^4+n+1\leq 2\cdot 3^4(n+1)^4,
\end{equation}
where in the last inequality we used $\x+\epsilon_n\leq \frac{1}{2}(n+1)+\frac{3}{2}\sqrt{n+1}+\frac{1}{2}\leq 3(n+1)$. Moreover, we have
\begin{equation}
\inf_{x\in\supp\sigma_n} (-V_o)^{\frac{1}{2}}=|p_{\z}|\sqrt{\x^4+1}\geq \frac{1}{4}|p_{\z}|(n+1)^2,
\end{equation}
where we applied $\x\geq \frac{1}{2}(n+1)$. We can now estimate
\begin{align}
&\int \left|\frac{V_1}{(-V_o)^{\frac{1}{2}}}\right|dx\leq 2\sum_{n=0}^\infty \epsilon_n
\frac{p_{\z}^2\sup_{x\in \supp\sigma_n} \sigma_n}{\inf_{x\in\supp\sigma_n} (-V_o)^{\frac{1}{2}}}\leq 2\sum_{n=0}^\infty \epsilon_n\frac{2\cdot 3^4|p_{\z}|(n+1)^4}{\frac{1}{4}(n+1)^2}.
\end{align}
As $\epsilon_n$ satisfy
condition \eqref{def-epsn} we see that the integral is finite.
\end{enumerate}
The assumptions of Lemma \ref{lm-pr} are satisfied and the operator $H_1$  has both deficiency indices equal to $2$. Finally, the same is true for the operator $\dA_{p_y,p_{\z},p_\eta}$ that differs by a constant from $H_1$, see \eqref{eq:H_1}.
\end{proof}

\subsection{Non-self-adjointness of the full operator}\label{sec:non}

We will consider the operator $\widehat{\dA}$ in the partial Fourier representation.
We showed, that $\dA_{p_y,p_{\z},p_\eta}$  has both deficiency indices equal to $2$. Thus, every distributional  solution of 
\begin{equation}
(\dA^\dagger_{p_y,p_{\z},p_\eta} +i)\psi=0,
\end{equation}
is square integrable and in fact such solutions are smooth. Let us define 
\begin{equation}
\psi_{p_y,p_{\z},p_\eta}\in \LL,
\end{equation}
as the unique smooth function satisfying
\begin{enumerate}
\item $(\dA_{p_y,p_{\z},p_\eta}^\dagger +i)\psi_{p_y,p_{\z},p_\eta}=0$,
\item the value of the function at $0$ satisfies 
\begin{equation}
\psi_{p_y,p_{\z},p_\eta}(0)=1,
\end{equation}
\item the derivative at $0$ satisfies 
\begin{equation}
\psi_{p_y,p_{\z},p_\eta}'(0)=0.
\end{equation}
\end{enumerate}
This function is symmetric
\begin{equation}
\forall_{x\in\R}\,\psi_{p_y,p_{\z},p_\eta}(x)=\psi_{p_y,p_{\z},p_\eta}(-x)
\end{equation}
because of the symmetry $x\rightarrow -x$ of the operator $\dA_{p_y,p_{\z},p_\eta}$. It is also nontrivial.

\begin{lm}\label{lm:measure}
The function $\R^3\ni (p_y,p_{\z},p_\eta)\rightarrow\|\psi_{p_y,p_{\z},p_\eta}\|\in \R\cup\{\infty\}$ is measurable.  It is also finite and positive for $p_{\z}\not=0$.
\end{lm}

\begin{proof}
The norms 
\begin{equation}
\|\psi_{p_y,p_{\z},p_\eta}\|_{[-L,L]}^2=\int_{-L}^Ldx\, |\psi_{p_y,p_{\z},p_\eta}|^2
\end{equation}
depend continuously on $p_y,p_{\z},p_\eta$ because functions $\psi_{p_y,p_{\z},p_\eta}$ do. Let us notice that 
\begin{equation}
\|\psi_{p_y,p_{\z},p_\eta}\|=\sup_{L} \|\psi_{p_y,p_{\z},p_\eta}\|_{[-L,L]},
\end{equation}
thus it is measurable. It is also finite everywhere except for $p_{\z}=0$. Positivity follows from positivity
\end{proof}

In fact, using proof of Lemma \ref{lm-8} one can show that this function is smooth for $p_{\z}\not=0$, but we will not need this fact.
Let us define 
\begin{equation}
U_K=\{(p_y,p_{\z},p_\eta)\in[1,2]^3\colon \|\psi_{p_y,p_{\z},p_\eta}\|<K\}.
\end{equation}
It is a measurable set (by Lemma \ref{lm:measure}) and the standard Lebesgue measure $\mu(U_K)\leq 1$. We can choose $K>0$ such that the Lebesgue measure $\mu(U_K)>0$, thus
\begin{equation}
\psi(x,p_y,p_{\z},p_\eta)=1_{U_K}(p_y,p_{\z},p_\eta)\psi_{p_y,p_{\z},p_\eta}(x)\in \HF
\end{equation}
is non-trivial in $\HF$ ($\|\psi\|_{\wedge}\not=0$). Here $1_{U_K}$ denotes characteristic function of the set $U_K$.

We will now show, that $\psi$ is perpendicular to the space \eqref{eq:imd}. Let us suppose that $\phi\in \widehat{D}$ (the Fourier transformed domain of $\dA$ operator) then from Lemma \ref{lm:F}
\begin{enumerate}
\item $\phi\in C^\infty(\R^4)$,
\item there exists $L>0$ such that $\supp \phi\subset[-L,L]\times \R^3$.
\end{enumerate}
Applying Fubini theorem and integrating by parts in $x$ we get using Lemma \ref{lm:F} and properties of functions $\psi_{p_y,p_{\z},p_\eta}$
\begin{align}
&\left\langle\psi,\widehat{\dA}\phi\right\rangle_{\hspace{-2pt}\wedge}
=\int_{U_K} dp_ydp_{\z}dp_\eta\ \int_{-2L}^{2L} dx\ \overline{\psi_{p_y,p_{\z},p_\eta}}\dA_{p_y,p_{\z},p_\eta}\phi=\nonumber\\
&=\int_{U_K} dp_ydp_{\z}dp_\eta\ \int_{-2L}^{2L} dx\  \overline{\dA_{p_y,p_{\z},p_\eta}^\dagger\psi_{p_y,p_{\z},p_\eta}}\ \phi=\nonumber\\&=i\int_{U_K} dp_ydp_{\z}dp_\eta\ \int_{-2L}^{2L} dx\  \overline{\psi_{p_y,p_{\z},p_\eta}}\ \phi=i\langle \psi,\phi\rangle_{\hspace{-1pt}\wedge}.
\end{align}
This shows that $\psi$ is perpendicular to the space \eqref{eq:imd}, thus
\begin{equation}
\psi\in \ker (\widehat{\dA}^\dagger+ i)
\end{equation}
The operator $\dA$ is not essentially self-adjoint on the domain $C^\infty_0(M_V)$ and we proved theorem \ref{thm:2}.

\section{Summary}

A globally hyperbolic and geodesically complete metric does not necessary have essentially self-adjoint d'Alembert and Klein-Gordon operators, thus one should be careful with objects like $e^{it\dA}$ or $(\dA+m^2\pm i\epsilon)^{-1}$ even in physically reasonable spacetimes. It is an open problem if some additional assumptions like Einstein's equations with chosen energy condition or adding to the operator the non-minimal coupling term \cite{Birrell} (proportional to the Ricci scalar)  would change the result. 

\section*{Acknowledgements}
We would like to thank Jan Derezi{\'n}ski and Daniel Siemssen for useful discussions. We would also like to express gratitude for anonymous referees of the previous version of the manuscript for pointing out various shortcomings.


\end{document}